\documentclass[a4paper,onecolumn,11pt,unpublished]{quantumarticle}
\pdfoutput=1
\usepackage[utf8]{inputenc}
\usepackage[english]{babel}
\usepackage[T1]{fontenc}
\usepackage{amsmath,amsthm,amsfonts,thm-restate}
\usepackage[margin=0.68in]{geometry}
\usepackage[hypertexnames=false,colorlinks,urlcolor=blue,citecolor=blue]{hyperref}
\usepackage[dvipsnames]{xcolor}
\usepackage{physics}

\newtheorem{definition}{Definition}
\newtheorem*{theorem*}{Theorem}  

\usepackage{tikz}
\usetikzlibrary{quantikz}
\usepackage{lipsum}

\usepackage{soul}

\begin{document}

\title{Quantum Approximate Counting with Additive Error: Hardness and Optimality}

\author{Mason L. Rhodes}
\affiliation{Center for Computing Research, Sandia National Laboratories, Albuquerque, NM, 87185, USA}
\affiliation{Center for Quantum Information and Control, University of New Mexico, Albuquerque, NM, 87131, USA}
\affiliation{Department of Physics and Astronomy, University of New Mexico, Albuquerque, NM, 87131, USA}
\author{Sam Slezak}
\affiliation{Information Sciences, Los Alamos National Laboratory, Los Alamos, NM,87544, USA}

\author{Anirban Chowdhury}
\affiliation{IBM Quantum, IBM T.\ J.\ Watson Research Center, Yorktown Heights, NY, 10598, USA}
\author{Yi\u{g}it Suba\c{s}\i}
\affiliation{Information Sciences, Los Alamos National Laboratory, Los Alamos, NM,87544, USA}

\maketitle


\begin{abstract}

Quantum counting is the task of determining the dimension of the subspace of states that are accepted by a quantum verifier circuit. It is the quantum analog  of counting the number of valid solutions to {\sf NP} problems - a problem well-studied in theoretical computer science with far-reaching implications in computational complexity. 
The complexity of solving the class {\sf \#BQP} of quantum counting problems, either exactly or within suitable approximations, is related to the hardness of computing many-body physics quantities arising in algebraic combinatorics. Here, we address the complexity of quantum approximate counting under additive error. 
First, we show that computing additive approximations to {\sf \#BQP} problems to within an error exponential in the number of witness qubits in the corresponding verifier circuit is as powerful as polynomial-time quantum computation. Next, we show that returning an estimate within error that is any smaller is {\sf \#BQP}-hard. Finally, we show that additive approximations to a restricted class of {\sf \#BQP} problems are equivalent in computational hardness to the class {\sf DQC}$_1$. Our work parallels results on additively approximating {\sf \#P} and {\sf GapP} functions.
\end{abstract}


\section{Introduction}\label{sec:introduction}

Given a decision problem with polynomial-time verifiable solutions, a natural question to ask is how many solutions, i.e., accepting witnesses, exist. When the verifier is classical, this question essentially defines the counting complexity class {\sf \#P} \cite{VALIANT1979189}---the class of functions that can be expressed as the number of accepting witnesses to {\sf NP} problems. These counting problems appear in combinatorics, graph theory, and statistical physics but are known to be extremely hard to solve. In fact, a result of Toda \cite{Toda1991PPIA} implies that a polynomial number of queries to a {\sf\#P} oracle can solve any problem in the entire polynomial hierarchy. 

The class {\sf\#P} precisely characterizes the complexity of counting problems such as computing the permanent of binary matrices, counting the number of perfect matchings in a graph \cite{VALIANT1979189}, and computing partition functions of certain classical spin systems \cite{jerrum1993ising,bulatov2005partition}. 
On the other hand, a number of important problems in, e.g., algebraic combinatorics \cite{buergisser2008complexity} and quantum computing \cite{Brown_2011} are polynomial-time reducible to {\sf \#P} but are not {\sf \#P} problems. Another important counting class is {\sf GapP} \cite{fortnow1998complexitylimitationsquantumcomputation}, defined as the closure of {\sf \#P} under subtraction, which is known to contain problems such as computing multiplicities arising in the representation theory of the symmetric group \cite{buergisser2008complexity}.

The counting class {\sf \#BQP} was defined to formalize the notion of counting the number of quantum witnesses to promise problems that allow efficient verification on quantum computers \cite{Brown_2011, shi2009note}. This class contains quantum-inspired counting problems such as computing the density of states of quantum Hamiltonians \cite{Brown_2011} and computing representation theoretic multiplicities such as Kronecker coefficients \cite{ikenmeyer2023quantumkronecker,bravyi2024kronecker}.
Though {\sf \#P}, {\sf GapP}, and {\sf \#BQP} are equivalent under polynomial-time reductions \cite{Brown_2011,counting-O}, the subtle difference between them manifests directly in the hardness of approximation, specifically in approximating the respective quantities within a multiplicative error. 
Any {\sf \#P} function can be multiplicatively approximated in the third level of the polynomial hierarchy \cite{Stockmeyer1983TheCO}, a class of problems believed to be easier than ${\sf \# P}$, but approximating {\sf GapP} functions can be {\sf \#P}-hard in general \cite{goldberg2017complexityapproximatingcomplexvaluedising}. However, it is unknown if multiplicatively approximating {\sf \#BQP} relations is any easier even though this has been used to define a new class {\sf QXC} of quantum approximate counting problems \cite{bravyi2024kronecker,Bravyi_2022}.

It is thus interesting to ask if other notions of approximation also lead to problems of varying complexity. The study of \textit{additive} approximations to counting problems was initiated in Ref.~\cite{bordewich2009approximatecountingquantumcomputation}, motivated by the result that certain additive approximations to {\sf GapP} problems suffice to solve {\sf BQP} decision problems \cite{fortnow1998complexitylimitationsquantumcomputation}. There, it was shown that all {\sf \#P} and {\sf GapP} functions have efficient \textit{classical} additive-approximation algorithms up to a normalization factor exponential in the size of the descriptions of the solutions, with no difference in the complexity of obtaining the additive approximations for {\sf \#P} and {\sf GapP}. In quantum computation, additive-error approximations to some {\sf \#BQP} and {\sf QXC} problems including estimating the quantum density of states and quantum partition functions have been shown to be complete for the class {\sf DQC$_1$} \cite{brandao2008entanglement,Chowdhury2021PartitionDQC1,cade2017quantum}, the set of decision problems that can be solved in a model of quantum computation where the system is initialized with only one pure qubit \cite{knill1998power,shepherd2006computation}.

In this work, we study the complexity of obtaining additive-error approximations to {\sf \#BQP} problems and its connection to the complexity classes {\sf BQP} and {\sf DQC$_1$}. 
We first prove that returning additive-error approximations to {\sf \#BQP} functions, normalized by the dimension of the witness subspace, is polynomial time equivalent to {\sf BQP}, and thus captures the difficulty of polynomial-time quantum computation. Next, we show that this level of approximation is likely optimal as the problem of computing more accurate approximations becomes {\sf \#BQP}-hard. We then demonstrate that a weaker additive approximation is achievable via a polynomial-time classical algorithm. Next we present results that elucidate the connection between {\sf DQC$_1$} and {\sf \#BQP}. We show that {\sf DQC$_1$} is equivalent to computing additive approximations to a subset of {\sf\#BQP} problems---those with verifier circuits containing at most a logarithmic number of ancilla qubits.

\section{Background}\label{sec:background}

\subsection{Classical Counting Complexity}\label{sec:classical-counting}

Classical counting complexity, including classes such as {\sf \#P} and {\sf GapP}, characterizes the difficulty of computing the number of accepting witnesses to {\sf NP} problems. A problem is in {\sf NP} if given a ``yes'' instance $x$ there is an efficiently generated polynomial-time verification procedure and at least one witness $y$ such that the procedure accepts on $(x,y)$, while for ``no'' instances, no such witness exists. 

We choose to define both {\sf NP} and {\sf \#P} via relations $R\subset \{0,1\}^*\times\{0,1\}^* $ between the set of inputs and the set of possible witnesses such that there exists a polynomial-time Turing machine that decides the proposition $R(x,y)$. Introducing the set of accepting witnesses of length bounded by a polynomial $p$ for input $x$: 
\begin{align}\label{eq:R-set-classical-accepting-witnesses}
    R(x) = \{y:R(x,y),|y|\leq p(|x|)\}\,,
\end{align}
solving the {\sf NP} decision problem is to decide whether $|R(x)|>0$ while solving the {\sf\#P} problem is to return $|R(x)|$. The class {\sf GapP} is then defined as the set of functions that are the difference of two {\sf \#P} functions.

\subsection{Acceptance Operators and {\sf QMA}}\label{sec:accept-op-qma}
A quantum generalization of {\sf NP} is the class {\sf QMA}, which roughly corresponds to {\sf NP} except the witnesses are quantum states and the verification procedure is a quantum circuit. We will define {\sf QMA}, and later the quantum analogue of {\sf \#P}, known as {\sf \#BQP}, via properties of \textit{acceptance operators} $\mathcal V$ \cite{Bravyi_2022,Brown_2011} which generalize the relations $R$ in the previous section. We define them via a quantum circuit $V$ that acts on a Hilbert space $\mathcal H_I\otimes\mathcal H_A\otimes\mathcal H_W$ where $\mathcal H_W$ is the subspace spanned by the quantum witnesses, $\mathcal H_A$ is the ancilla subspace needed for the verification procedure, and $\mathcal H_I$ is the subspace of the (classical) input. Given the circuit $V$ requiring $a$ ancilla qubits and a classical input bit string $x$, we define the acceptance operator $\mathcal V_x$ via
\begin{align}\label{eq:gen-acceptance-op}
    \mathcal V_x = (\bra{x}\otimes\langle 0^a|\otimes\mathbb{I}_W)V^\dagger \Pi_{\text{out}}V(\ket{x}\otimes |0^a\rangle\otimes\mathbb{I}_W)\,,
\end{align}
where $\Pi_\text{out}$ is a projector onto the $\ketbra{1}{1}$ state of the first qubit. Using these operators, we can define the complexity class {\sf QMA}.
\begin{definition}[{\sf QMA}]\label{def:QMA}
    Let $A = (A_{\text{YES}},A_{\text{NO}})$ be a promise problem and let $c,s: \mathbb N \to [0,1]$ be functions satisfying $c(n)-s(n) = \Omega(\text{poly}(n)^{-1})$. Then $A\in {\sf QMA}$ if and only if there exists a polynomial-time uniform family of quantum circuits $\{{V^{(n)}}:n\in \mathbb N\}$ acting on $w(n) = O(\text{poly}(n))$ witness qubits, $a(n) = O(\text{poly}(n))$ ancilla qubits, and $n$ input qubits such that:
    \begin{enumerate}
        \item if $x\in A_{\text{YES}}$, then $||\mathcal V_x||\geq c(|x|)$\,,
        \item if $x\in A_{\text{NO}}$, then $||\mathcal V_x||\leq s(|x|)$\,,
    \end{enumerate}
    where $\mathcal V_x$ is the acceptance operator of Eq.~\eqref{eq:gen-acceptance-op} corresponding to $V^{(|x|)}$ on input $x$.  
\end{definition}
Often in the literature the functions $c$ and $s$ are taken to be the constants $2/3$ and $1/3$, respectively. We choose this definition as it motivates the discussion of the promise gap $c(n)-s(n)$, which becomes important when defining {\sf \#BQP} in Section \ref{sec:sharp-BQP} below.

\subsection{{\sf \#BQP}}\label{sec:sharp-BQP}

Some care needs to be taken in defining the counting class associated to {\sf QMA}, known as {\sf \#BQP}. Intuitively, one would like to generalize the notion of computing the size of the set $R(x)$ in Eq.~\eqref{eq:R-set-classical-accepting-witnesses} to the quantum case. This is not straightforward, as the number of possible quantum witnesses is uncountably infinite. Instead, a counting problem can be formulated as determining the dimension of the subspace of good witnesses. 
This subspace will be defined via the spectral properties of the acceptance operator $\mathcal V_x$. 
We first define the maps $N_{\geq a} $ from Hermitian operators to integers via
\begin{align}
    N_{\geq a}(H)  = |\{b:b\in\text{spec}(H),b\geq a\}|\,.
\end{align}
That is, they return the number of eigenvalues of $H$ that are greater than or equal to $a$. One can then attempt to define {\sf\#BQP} via counting functions such that given a circuit $V$ and an input $x$ corresponding to an acceptance operator $\mathcal V_x$, return $N_{\geq c}(\mathcal V_x)$. This ends up being too powerful, as the ability to return this would allow one to solve problems in the variant of {\sf QMA} where the promise gap, $c-s$, is allowed to be exponentially small. This instead characterizes the class {\sf PSPACE} \cite{fefferman2016quantummerlinarthurexponentially}. Thus, one must be careful with how the promise gap is handled.   

The original formulation of ${\sf \#BQP}$ \cite{Brown_2011,shi2009note} circumvents this by imposing a promise that acceptance operators $\mathcal V_x$ corresponding to allowed problem instances $x$ do not have any eigenvalues in the interval $(c,s)$. Conveniently, this also guarantees that the output of any ${\sf \#BQP}$ function is unique but has the consequence that proving inclusion of problems in ${\sf \#BQP}$ requires adding extra promises. In particular, Ref.~\cite{Brown_2011} showed that the density of states is a ${\sf \#BQP}$-complete problem but only for Hamiltonians that satisfy a potentially nonphysical promise on their spectrum. 
As Ref.~\cite{Brown_2011} notes, there is an alternative approach to defining ${\sf \# BQP}$ where one is allowed to ``miscount'' eigenvalues in the interval $(c,s)$. This was later generalized to define the notion of relative-error approximations to ${\sf \# BQP}$ problems and to show that this is equivalent to problems such as computing relative-error approximations to quantum partition functions under polynomial-time reductions~\cite{Bravyi_2022}. Here, we choose to work with the definition of ${\sf \#BQP}$ which allows for miscounting, as it does not require additional assumptions on the associated verifier circuits. 

In order to rigorously define {\sf\#BQP} problems with the allowed miscounting in the promise gap, we again define them as relations $r\subset J\times \mathbb Z^+$ between a set of inputs $J\subset\{0,1\}^*$, and the set of positive integers $\mathbb Z^+$. Note that the use of $r$ is to distinguish them from the relations $R$ used in the definition of {\sf \#P}. The relations considered here include standard functions as a specific instance, satisfying the property that for every $x
\in J$ there is exactly one $y\in\mathbb Z^+$ related to it.  We denote the set of outputs related to the input $x\in J$ via
\begin{align}
    r(x) \equiv \{y: r(x,y) \}\,.
\end{align}
A \textit{relational problem} is then a problem where given an input $x\in J$, the goal is to return any $y\in r(x)$.
With this in mind we can now define {\sf \#BQP}. 
\begin{definition}[{\sf \#BQP}]\label{def:sharp-BQP}
    A relation $r\subset J \times \mathbb Z^+$ is in {\sf\#BQP} if and only if there exists a uniformly-generated family of quantum circuits $V^{r,(n)}$ and functions $c^r,s^r: \mathbb N \to [0,1]$ satisfying $c^r(n)-s^r(n) = \Omega(\text{poly}(n)^{-1})$ such that 
    $$
       r(x) = [N_{\geq c^r(|x|)}(\mathcal V_x^{r})..N_{\geq s^r(|x|)}(\mathcal V_x^{r})]\,, 
    $$
    where $\mathcal V_x^{r}$ is the acceptance operator corresponding to the circuit $V^{r,(|x|)}$ on input $x$.
\end{definition}
Here the notation $[a..b]$ denotes the set of integers between and including $a$ and $b$. For the purposes of this work, certain parameters of the circuits $V^{r,(n)}$ will be very important. 
These parameters are defined as follows: Given a family of verifier circuits $\{V^{r,(n)}\}_{n\in \mathbb N}$ in the $\{\text{Toffoli}, H, S\}$ gate set where $r$ denotes the family and $n$ denotes the input size, define the parameters $w_r,a_r,t_r,h_r:\mathbb Z^+ \to \mathbb Z^+$ via 
\begin{itemize}
    \item $w_r(n)$ as the number of qubits in the witness register,
    \item $a_r(n)$ as the number of qubits in the ancilla register,
    \item $t_r(n)$ as the total number of gates,
    \item $h_r(n)$ as the total number of $H$ gates. 
\end{itemize}

\subsection{{\sf BQP} and {\sf DQC$_1$}}\label{sec:BQP-DQC1}

We now recall the standard definition of {\sf BQP}, which is the class of decision problems that can be solved by polynomial-time quantum computation.

\begin{definition}[{\sf BQP}]
    Let $A = (A_{\text{YES}},A_{\text{NO}})$ be a promise problem and let $c,s: \mathbb N \to [0,1]$ be functions satisfying $c(n)-s(n) = \Omega(\text{poly}(n)^{-1})$. Then $A\in{\sf BQP}$ if and only if there exists a polynomial-time uniform family of quantum circuits $\{V^{(n)}:n\in \mathbb N\}$ acting on $a(n) = O(\text{poly}(n))$ ancilla qubits and $n$ input qubits such that:
    \begin{enumerate}
        \item if $x\in A_{\text{YES}}$, then $\bra{x, 0^a}V^{(|x|)\dagger}\Pi_{\textup{out}}V^{(|x|)}\ket{x, 0^a} \geq c(|x|)$\,,
        \item if $x\in A_{\text{NO}}$, then $\bra{x, 0^a}V^{(|x|)\dagger}\Pi_{\textup{out}}V^{(|x|)}\ket{x, 0^a} \leq s(|x|)$\,.
    \end{enumerate}
\end{definition}

Another computational complexity class that we will be working with is {\sf DQC$_1$}. Introduced in Ref.~\cite{knill1998power}, {\sf DQC}$_1$ is a limited quantum computational model in which only a single qubit can be prepared in a pure state while the rest of the qubits used in the computation are in the maximally mixed state. It is known that {\sf DQC$_1$} is equal to {\sf DQC$_{\log}$}~\cite{shepherd2006computation,shor2008estimating}, which allows for logarithmically many ancilla qubits to be prepared in a pure state. We use this fact to define {\sf DQC$_1$} in terms of acceptance operators $\mathcal D_x$ similar to how {\sf QMA} and {\sf BQP} were defined. Note that for {\sf DQC$_1$}, rather than employing polynomial-time uniform generation for the circuits as before, an alternative scheme must be used where a description of a circuit for each input $x$ is generated in polynomial time. (This is simply to avoid having to count the input $x$ among the qubits prepared in a pure state.) Thus given a circuit $D^{(x)}$ on $a$ ancilla qubits and $w$ witness qubits, define the corresponding acceptance operator as
\begin{align}\label{eq:dqc1-acceptance-op}
    \mathcal D_x = (\langle 0^a|\otimes \mathbb I_W)D^{(x)\dagger} \Pi_{\text{out}}D^{(x)}(|0^a\rangle\otimes \mathbb I_W)\,. 
\end{align}
With this clarified, we are now in a position to define {\sf DQC$_1$}.
\begin{definition}[{\sf DQC}$_1$]\label{def:dqc1}
    Let $A = (A_{\text{YES}},A_{\text{NO}})$ be a promise problem and let c,s: $\mathbb N\to [0,1]$ be functions satisfying $c(n)-s(n) =  \Omega(\text{poly}(n)^{-1})$. Then $A\in {\sf DQC}_1$ if and only if for every input $x\in A$ there exists a polynomial-time-generated quantum circuit $D^{(x)}$ with corresponding acceptance operator $\mathcal D_x$ as in Eq.~\eqref{eq:dqc1-acceptance-op} acting on  $w(|x|) = O(\text{poly}(|x|))$ witness qubits and $a(|x|) = O(log(|x|))$ ancilla qubits such that:
\begin{enumerate}
    \item If $x\in A_{\text{YES}}$, then $\frac{\textup{Tr}[\mathcal D_x]}{2^{w(|x|)}}\geq c(|x|)$\,,
    
    \item If $x\in A_{\text{NO}}$, then $\frac{\textup{Tr}[\mathcal D_x]}{2^{w(|x|)}}\leq s(|x|)$\,.
\end{enumerate}
\end{definition}

\subsection{Additive-Error Approximations}\label{sec:add-error-approx}
In Ref.~\cite{bordewich2009approximatecountingquantumcomputation} it was shown that
all {\sf\#P}
and {\sf GapP} functions can be additively approximated to accuracies related to the length of the witnesses of the associated {\sf NP} problems defining them. In particular, given a {\sf \#P} or {\sf GapP} function $f$ with witnesses on $w(|x|)$ bits for inputs of size $|x|$, the authors show that there exists a classical probabilistic algorithm on input $x$ that outputs an estimate $\tilde f(x)$ satisfying
\begin{align}
    |f(x)-\tilde f(x)|\leq \epsilon u(|x|)\,,
\end{align}
to a normalization $u(|x|) = 2^{w(|x|)}$, with probability $1-\delta$, and that runs in time polynomial in $1/\epsilon$ and $\log(1/\delta)$. The fact that this level of additive approximation yields an equivalent complexity for both {\sf \#P} and {\sf GapP} follows immediately from another result of the paper, namely that additive approximations with a normalization $u$ are closed under subtraction. 

In this work we study whether the complexity of analogous additive approximations to {\sf\#BQP} relations are of a different complexity. In order to do so, we adapt the definitions of additive approximations from
Ref.~\cite{bordewich2009approximatecountingquantumcomputation} to relational problems, as well as framing them as computational problems rather than algorithms.
\begin{definition}[Additive Error Approximation]\label{def:additive-error-approximations}
    Given a relation $r\subset J \times \mathbb Z^+$ and a normalization $u:\mathbb Z^+ \to \mathbb R^+$, the problem of returning an additive-error approximation to $r$ with normalization $u$ to precision $\epsilon$, is to given an input $x\in J$, return a number $\tilde r(x,u)$ satisfying
    $$
    \min\{|\tilde r(x,u) - y|:y\in r(x,\cdot)\}\leq \epsilon u(|x|)\,.
    $$
    where $\epsilon=\Omega(\textup{poly}(|x|^{-1}))$.
\end{definition}
Intuitively this just means that the estimate needs to be close to at least one of the outputs $y$ related to an input $x$.

\section{Summary of Results}\label{sec:summary-of-results}

\subsection{The Complexity of Additive-Error Approximations to {\sf \#BQP}}\label{sec:complexity-add-error-apprx-sharp-bqp}
Here we present results related to additively approximating {\sf \#BQP} relations. The proofs for all of the theorems in this section are contained in Section \ref{sec:proofs-add-error-apprx-to-sharp-bqp}. For our first result, we demonstrate that returning additive approximations to {\sf \#BQP} functions with normalization factors $u$ scaling exponentially in the number of qubits in the witness space is polynomial-time equivalent to {\sf BQP}.

\begin{restatable}{thm}{BPPAisBQP}\label{thm:BPPA-BQP}
 Let $A$ be an oracle that given as input a {\sf \#BQP} relation $r$ with verifier circuits $V^{r,(n)}$, an input string $x$, and some $\epsilon=\Omega(\text{poly}(|x|)^{-1})>0$, returns an additive-error approximation to $r(x)$ with $u(|x|) =2^{w_r(|x|)}$. Then {\sf BQP} $=$ {\sf BPP}$^A$.
\end{restatable}

The consequences of this theorem are two-fold. First, since the proof in Section~\ref{sec:proofs-add-error-apprx-to-sharp-bqp} shows that the oracle $A$ can be simulated by a quantum subroutine, in analogy to the results in Ref.~\cite{bordewich2009approximatecountingquantumcomputation}, 
there are efficient quantum algorithms for returning additive-approximations for all {\sf\#BQP} functions, but with the normalization $u$ now being determined by the number of qubits rather than the number of bits describing the witnesses. On the other hand, this result implies that returning additive approximations to {\sf\#BQP} functions to the stated $u$ is {\sf BQP}-hard, so it is unlikely that they can be efficiently computed classically. We next show that this level of additive approximation is, in a sense, optimal. 

\begin{restatable}{thm}{bestapprox}\label{thm:better-u-is-sharpBQP-hard}
    The problem of returning additive-error approximations for all {\sf \#BQP} relations $r$ with corresponding acceptance operators  $\mathcal V_x^{r}$ on input $x$ to normalization $u(|x|) =  2^{cw_r(|x|)}$ for positive $c<1$ satisfying $1-c = \Omega(\text{poly}(|x|
    )^{-1})$ is {\sf\#BQP}-hard. 
\end{restatable}

 Given the fact that exact counting classes are equivalent under polynomial-time reductions \cite{Brown_2011,counting-O}, and as shown in Ref.~\cite{bordewich2009approximatecountingquantumcomputation} there are efficient classical additive-error approximations to all {\sf \#P} and {\sf GapP} problems to some normalization $u$ there must exist classical additive approximations to {\sf \#BQP} relations. The key point that makes this consistent with the previous {\sf BQP}-hardness result is that the reductions between the classes have a profound effect on the size of the witnesses, leading to a very different normalization $u$ between the quantum and classical approximation algorithms.  We use the reduction in Ref.~\cite{Brown_2011} that maps a {\sf \#BQP} relation to a {\sf GapP} function with a witness size that results in a normalization $u$ which is significantly larger than the one in Theorem~\ref{thm:BPPA-BQP}. A downside of the construction is that the achievable normalization likely depends on the gate set used.

\begin{restatable}{thm}{classicalapprox}\label{thm:classical-approx-sharpBQP}
 For a {\sf\#BQP} relation $r$ with input $x$, let $h_r \equiv h_r(|x|)$ be number of Hadamard gates in the circuit $V^{r,(|x|)}$ using the $\{\text{Toffoli},H,S\}$ gate set and let $t_r \equiv t_r(|x|)$ be the circuit depth and $a_r \equiv a_r(|x|)$ be the number of ancilla qubits. Then classical additive-error approximations exist with normalization factor $u(|x|)=2^{2t_r(w_r+a_r) - (a_r +1) - h_r}$.
\end{restatable}
Note that this is in general not a particularly useful approximation, as $u$ is typically much greater than the range of values that the {\sf\#BQP} function can return in the first place.

\subsection{The Relationship between Additive-Error Approximations to {\sf\#BQP} and {\sf DQC$_1$}}\label{sec:add-apprx-sharp-bqp-dqc1}
 Motivated by a collection of results showing that additive approximations to problems related to {\sf\#BQP}, including the log-local quantum partition function and the log-local quantum density of states problems, are {\sf DQC$_1$}-complete \cite{brandao2008entanglement,Chowdhury2021PartitionDQC1,gyurik2022towardsquantum}, we analyze the general connection between additive approximations to {\sf \#BQP} relations and the complexity class {\sf DQC$_1$}. Although the density of states problem is {\sf \#BQP}-complete, the reductions used affect the levels of additive error achievable, so conclusions are not immediate about the relationship between additive approximations to general {\sf\#BQP} relations and {\sf DQC$_1$}. We unify these results by showing that additive approximations to a subset of {\sf \#BQP} relations, those requiring only logarithmically many ancilla qubits, have the same computational power as {\sf DQC$_1$}. The proofs of this theorem is presented in Section \ref{sec:proofs-DQC1-sharpBQP}.

\begin{restatable}{thm}{BPPAisDQC}\label{thm:BPPA-DQC1}
 Let $A$ be an oracle that given as input a {\sf \#BQP} relation $r$ with verifier circuits $V^{r,(n)}$, where $a_r(n) = O(\log(w_r(n)))$, an input string $x$, and some $\epsilon=\Omega(\text{poly}(|x|)^{-1})>0$,  returns an additive-error approximation to $r(x)$ with $u(|x|) =2^{w_r(|x|)}$. Then {\sf DQC}$_1$ $=$ {\sf BPP}$^A$.
\end{restatable}

Theorem~\ref{thm:BPPA-DQC1} immediately implies the {\sf DQC$_1$} completeness results of contained in \cite{brandao2008entanglement,gyurik2022towardsquantum}, as the verifier circuits for calculating the density of states only require logarithmically many ancillas, and polynomial-time reductions can be used to show that this result implies the {\sf DQC$_1$} completeness of computing additive error approximations to partition functions in \cite{Chowdhury2021PartitionDQC1}.

\section{Proofs of Results}\label{sec:proofs}

\subsection{Additive-Error Approximations to {\sf\#BQP}}\label{sec:proofs-add-error-apprx-to-sharp-bqp}

Prior to proving Theorem~\ref{thm:BPPA-BQP} we will introduce an alternative definition of {\sf \#BQP} and a {\sf BQP}-complete problem that will make the reductions easier. The alternative version of {\sf \#BQP} below allows the promise gap to explicitly depend on the input.

\begin{definition}[{\sf \#BQP$_{(c^r,s^r)}$}]\label{def:sharp-BQP-new}
    A relation $r\subset J \times \mathbb Z^+$ is in {\sf\#BQP} if and only if there exists a uniformly-generated family of quantum circuits $V^{r,(n)}$ and a polynomial $q^r$ such that
    $$
       r(x) = [N_{\geq c^r(x)}(\mathcal V_x^{r})..N_{\geq s^r(x)}(\mathcal V_x^{r})]\,, 
    $$ 
where $c^r,s^r:\{0,1\}^*\to [0,1]$ are efficiently-computable functions satisfying $c^r(x)-s^r(x)>1/q^r(|x|)$.
\end{definition}
We show in Appendix~\ref{app:sharp-bqp-equiv} that Definition~\ref{def:sharp-BQP-new} is equivalent to Definition~\ref{def:sharp-BQP}. The {\sf BQP}-complete problem we introduce is that of estimating the average acceptance probability of verifier circuits. We prove its {\sf BQP}-completeness in Appendix~\ref{app:avg-accp-verifer-bqp-complete}.

\begin{definition}[Average Accept Probability of Verifier Circuits]\label{def:avg-accp-verifier-circ}
Given functions $c,s: \mathbb N \to [0,1]$ satisfying $c(n)-s(n) = \Omega(\text{poly}(n)^{-1})$ and a polynomial-time uniform family of circuits $\{V^{(n)}:n\in\mathbb N\}$, each acting on $w(n) = O(\text{poly}(n))$ witness qubits, $a(n) = O(\text{poly}(n))$ ancilla qubits, and $n$ input qubits, the Average Accept Probability of Verifier Circuits problem is to, given an input $x$, decide whether 
\begin{enumerate}
    \item $\frac{\text{Tr}[\mathcal V_x]}{2^{w(|x|)}}\geq c(|x|)$\,,
    \item or $\frac{\text{Tr}[\mathcal V_x]}{2^{w(|x|)}}\leq s(|x|)$\,,
\end{enumerate}
promised one of these is the case, where $\mathcal V_x$ is the acceptance operator of Eq.~\eqref{eq:gen-acceptance-op} corresponding to $V^{(|x|)}$ on input $x$.   
\end{definition}
Using these two definitions, we are now ready to prove Theorem~\ref{thm:BPPA-BQP}, restated here for convenience.

\BPPAisBQP*

\begin{proof}
     In order to show {\sf BPP}$^A$ $\subseteq$ {\sf BQP}, we prove there exists a polynomial-time quantum subroutine which simulates $A$. Given a verifier circuit $V^{r,(|x|)}$ with completeness and soundness parameters $c$ and $s$, we construct a new circuit $\tilde{V}^{r,(|x|)}$ with amplified completeness and soundness parameters $\tilde{c}$ and $\tilde{s}$ to be chosen later. Note that such a circuit can always be constructed in polynomial time while preserving the dimension of the witness subspace~\cite{Marriott:2005a,Brown_2011}. The algorithm is then to generate $M=\text{poly}(|x|)$ uniformly random computational basis states for the witness subspace and send them through the amplified verifier circuit, then count how many accept and multiply this number by $2^{w_r(|x|)}/M$. Let $X_i$ be the indicator function that evaluates to $1$ if the $i$th bit string accepts. Then the estimator is
        \begin{align}\label{eq:q-add-approx-estimator}
            X = \frac{2^{w_r(|x|)}}{M} \sum_{i=1}^MX_i\,.
        \end{align}
        It can be shown that $\mathbb E[X] = \text{Tr}[\tilde{\mathcal{V}}_x^{r}]$ where $\tilde{\mathcal{V}}_x^{r}$ are the acceptance operators associated to the amplified verifier circuits $\tilde{V}^{r,(|x|)}$, as in Eq.~\eqref{eq:gen-acceptance-op}. In terms of the amplified completeness and soundness parameters, we have
        \begin{align}\label{eq:trace_estimator}
            N_{\geq c}(\mathcal V_x^{r}) -(1-\tilde c)2^{w_r(|x|)} \leq \text{Tr}[\tilde{\mathcal V}_x^{r}]\leq N_{\geq s}(\mathcal V_x^{r})+\tilde s2^{w_r(|x|)}\,.
        \end{align}
        The variance of $X$ is then given by
        \begin{align}
            \text{Var}[X] = \frac{1}{M}\text{Tr}[\tilde{\mathcal V}_x^{r}](2^{w_r(|x|)}-\text{Tr}[\tilde{\mathcal V}_x^{r}])\,.
        \end{align}
        From this it follows that
        \begin{equation}
        \begin{aligned}
            \text{Pr}(|X - \text{Tr}[\tilde{\mathcal V}_x^{r}]|\geq \epsilon 2^{w_r(|x|)})
            \leq & \frac{\text{Var}[X]}{\epsilon^2 2^{2w_r(|x|)}}\\
             = & \frac{1}{M\epsilon^2}\frac{\text{Tr}[\tilde{\mathcal{V}}_x^{r}]}{2^{w_r(|x|)}}\left(1-\frac{\text{Tr}[\tilde{\mathcal{V}}_x^{r}]}{2^{w_r(|x|)}}\right) \\
             \leq& \frac{1}{M\epsilon^2}\,.
        \end{aligned}
        \end{equation}
        Thus with constant failure probability $\delta\leq 1/4$ and $M= O(1/\epsilon^2)$ samples we can output a sample $\tilde X$ satisfying
        \begin{align}
            N_{\geq c}(\mathcal V_x^{r}) -(1-\tilde c +\epsilon)2^{w_r(|x|)} \leq\tilde X \leq  N_{\geq s}(\mathcal V_x^{r})+(\epsilon +\tilde s)2^{w_r(|x|)}\,.
        \end{align}
            Now, we choose $\tilde c \geq 1-1/\text{poly}(|x|)$ and $\tilde s \leq 1/\text{poly}(|x|)$ such that both $(1-\tilde c +\epsilon)$ and $\tilde s+\epsilon$ are less than $2\epsilon = \epsilon' = \Omega(\text{poly}(|x|)^{-1})$ so that the sample satisfies
            \begin{align}
                 N_{\geq c}(\mathcal V_x^{r}) -\epsilon'2^{w_r(|x|)} \leq\tilde X \leq  N_{\geq s}(\mathcal V_x^{r})+\epsilon' 2^{w_r(|x|)}\,.
                 \label{eq:normalized_trace_estimation_BQP}
            \end{align}
            The success probability can then be amplified to $1-\delta'$ with $\log(1/\delta')$ applications of the previous procedure and then taking the median of the $\log(1/\delta')$ estimates. We thus have a quantum algorithm that outputs an additive-error estimate with normalization $u(|x|)=2^{w_r(|x|)}$ that has a total runtime and size polynomial in $|x|$, $1/\epsilon'$ and $\log\left(1/\delta'\right)$. Therefore, {\sf BPP}$^A\subseteq$ {\sf BQP}. 

    On the other hand, to show {\sf BQP} $\subseteq$ {\sf BPP$^A$}, we  use the alternative version of {\sf \#BQP$_{(c^r,s^r)}$} from Definition~\ref{def:sharp-BQP-new}. Then querying $A$ for a {\sf \#BQP$_{(c^r,s^r)}$} relation $r$ on a polynomial number of inputs $x_i$ with normalizations $u_i=2^{w(|x_i|)}$, we show how this allows us to solve an instance of the Average Accept Probability of Verifier Circuits problem in Definition~\ref{def:avg-accp-verifier-circ} with parameters $c(|x|)$ and $s(|x|)$, a {\sf BQP}-complete problem (see Lemma~\ref{lem:avg-accept-verifer-circ-BQP-comp}). In particular, let $x\mapsto x_i = y_ix$ where $y_i$ is a bit string which encodes the correct completeness and soundness parameters to use for the corresponding input. Choosing a fixed length prefix-free encoding for $y_i$ ensures uniqueness of each concatenated bit string, and each $x_i$ has the same number of bits, which will be necessary to avoid repeated inputs. The completeness and soundness parameters, $c_i$ and $s_i$, are then determined by $y_i$.     

    Partition the interval $[0,1]$ into $M=\mathcal{O}(\text{poly}(|x|))$ subintervals $[s_i,c_i]$, each of size $\Omega(\text{poly}(|x|)^{-1})$, corresponding to the different completeness and soundness parameters specified by $y_i$. Specifically, we consider the subintervals $[s_i,c_i]$ where $c_i=(M-i)/M$ and $s_i=c_i-1/4M$ for $1\leq i\leq M-1$. Then querying $A$ for each problem input $x_i$ returns an estimate
    \begin{equation}
        \hat{n}_i=N_{\geq c_i}+\delta_i+\varepsilon_i\,,
    \end{equation}
    with $\delta_i\in[0,N_{[s_i,c_i]}]$ and $|\varepsilon_i|\in[0,\epsilon 2^{w(|x_i|)}]$. Letting $\hat{n}_0=0$, we define the two quantities $\hat{\chi}_i=\hat{n}_i-\hat{n}_{i-1}$ and $\Delta_i=c_i+1/2M$. Then we have
    \begin{equation}
    \begin{aligned}
        \sum_{i=1}^M \Delta_i\hat{\chi}_i&=\sum_{i=1}^M \Delta_i\left(N_{[c_i,c_{i-1}]}+\delta_i-\delta_{i-1}+\varepsilon_i-\varepsilon_{i-1}\right)\,, \\
        &=\sum_{i=1}^M \Delta_iN_{[c_i,c_{i-1}]}+\sum_{i=1}^{M-1}(\delta_i+\varepsilon_i)+\underbrace{\frac{1}{2M}(\delta_M+\epsilon_M)}_{=0}\,,
    \end{aligned}
    \end{equation}
    where the last term vanishes as a result of the fact that for the interval $[s_M,c_M]$, we need not query the oracle and instead always return $2^{w(|x|)}$, meaning $\delta_M,\epsilon_M=0$.

    This quantity approximates the trace of the associated acceptance operator $\mathcal{V}_{x}^{r}$ of the Average Accept Probability of Verifier Circuits problem. In particular, note that
    \begin{equation}
        |\Delta_i N_{[c_i,c_{i-1}]}-\text{Tr}_{[c_i,c_{i-1}]}(\mathcal{V}_{x}^{r})|\leq \frac{1}{2M}N_{[c_i,c_{i-1}]}\,,
    \end{equation}
    which then implies that 
    \begin{equation}
        \bigg|\sum_{i=1}^M\Delta_iN_{[c_i,c_{i-1}]}-\text{Tr}(\mathcal{V}_{x}^{r})\bigg|\leq \frac{2^{w(|x|)}}{2M}\,.
    \end{equation}
    Finally, we bound the error resulting from $\delta_i$ and $\varepsilon_i$ by noting that
    \begin{equation}
        \frac{1}{M}\sum_{i=1}^M \delta_i\leq \frac{2^{w(|x|)}}{M},\quad \text{and}\quad \frac{1}{M}\sum_{i=1}^M \varepsilon_i \leq \epsilon 2^{w(|x|)}\,,
    \end{equation}
    where the first inequality holds from the fact that $[s_i,c_i]\cap[s_j,c_j]=\emptyset$ for all $i\neq j$. In the second inequality we have used the fact that $2^{w(|x_i|)}=2^{w(|x|)}$ for all $i$, or rather, the size of the witness subspace for each input $x_i$ is the same size as that of $x$ even though the classical inputs may vary in size.
    Setting $\epsilon=1/M$, we add the errors from each contribution and find that
    \begin{equation}
        \bigg|\sum_{i=1}^M\Delta_i\hat{\chi}_i-\text{Tr}(\mathcal{V}_x^{r})\bigg|\leq \frac{2^{w(|x|)}}{M}+\frac{2^{w(|x|)}}{M}+\frac{2^{w(|x|)}}{2M}=\frac{5}{2}\frac{2^{w(|x|)}}{M}\,.
    \end{equation}
    Choosing $M>5/(c(|x|)-s(|x|))$ suffices to distinguish between the accepting and rejecting cases of Definition~\ref{def:avg-accp-verifier-circ}. Therefore, {\sf BQP} $\subseteq$ {\sf BPP$^A$}.
\end{proof}

Next we prove Theorem~\ref{thm:better-u-is-sharpBQP-hard}, adapting a proof from Ref.~\cite{stackexchange}, demonstrating that returning additive approximations to {\sf \#BQP} with better normalization factors is {\sf \#BQP}-hard.

\bestapprox*

\begin{proof}
Given an input $x$ to a {\sf\#BQP} relation $r$ with corresponding verifier circuit $V^{r,(|x|)}$ and completeness and soundness parameters $c^r(|x|)$ and $s^r(|x|)$, define a new verifier circuit $V^{'r,(|x|)} = V^{r,(|x|)}\otimes \mathbb I^{\otimes l}$ corresponding to another {\sf \#BQP} relation $r'$ where the witness subspace has been expanded by $l$ qubits where $l$ will be chosen later. It then follows that $N_{\geq c^r(|x|)}(\mathcal V_x^{'r}) = 2^l N_{\geq c^r(|x|)}(\mathcal V_x^{r})$ and $N_{\geq s^r(|x|)}(\mathcal V_x^{'r}) = 2^l N_{\geq s^r(|x|)}(\mathcal V_x^{r})$. Now, assume an additive error estimate $\tilde r'(x,u)$ to the {\sf \#BQP} instance corresponding to the circuit $V^{'r,(|x|)}$ can be returned with a normalization $u_{r'}(|x|) = 2^{c(w_r(|x|)+l)}$ to some inverse polynomial precision $\epsilon<1$. Then the estimate $\tilde r = \tilde r'(x,u)/2^l$ satisfies
\begin{align}
    \min\{|\tilde r - y|: y\in r(x,\cdot)\} \leq 2^{w_r(|x|)-(1-c)l}\,.
\end{align}
If we take $l>\frac{w_r(|x|)}{1-c}+1$ then the error is less than $1/2$ and by rounding to the nearest integer, a number in $r(x,\cdot)$ can always be returned. This takes $l = O\left(\frac{w_r(|x|)}{1-c}\right)$ extra qubits and thus $l$, and the size of $V^{'r,(|x|)}$, will be upper bounded by a polynomial in $|x|$  whenever $1-c = \Omega(\text{poly}(|x|)^{-1})$. 

\end{proof}

We additionally show that efficient classical algorithms for returning additive-error approximations to {\sf \#BQP} functions exist, though they can yield larger normalization factors.

\classicalapprox*

\begin{proof}
Given a {\sf \#BQP} relation $r$ and an input $x$ with a corresponding verifier circuit $V^{r,(|x|)}$, set $n = w_r(|x|)$ as the number of witness qubits, $a = a_r(|x|)$ as the number of ancilla qubits, $T = t_r(|x|)$ as the size of the circuit in the $\{H,S,\text{Toffoli}\}$ gate set, and $h = h_r(|x|)$ as the number of Hadamard gates. Let $V^{r,(|x|)}=g_T\cdots g_1$ in the chosen gate set and let $\mathcal V_x^{r}$ be the corresponding verifier operator.      
     First note that because of the completeness and soundness parameters we have that 
     \begin{align}
          N_{\geq c}(\mathcal V_x^{r}) - (1-c)2^n\leq\text{Tr}[\mathcal V_x^{r}] \leq N_{\geq s}(\mathcal V_x^{r}) +s2^n\,.
     \end{align}
Define the operators $q_i$ via $q_i = g_i$ if $g_i\in\{\text{Toffoli},S\}$ and $q_i = \sqrt{2}g_i$ if $g_i = H$ so that all of the matrix elements of $q_i$ are in $\{1,-1,i,-i\}$. Now expand the trace

\begin{equation}
\begin{aligned}
    \text{Tr}[\mathcal V_x^{r}] &=  \sum_{y} \bra{0^a,x}g_1^\dagger g_2^\dagger \cdots g_T^\dagger\ketbra{1,y}{1,y}g_Tg_{T-1}\cdots g_1\ket{0^a,x}\\
     &=  \sum_{y,z_j} \langle 0^a,x|g_1^\dagger|z_1\rangle\!\langle z_1| g_2^\dagger|z_2\rangle \cdots \langle z_{T-1}|g_T^\dagger|1,y\rangle\!\langle 1,y|g_T|z_T\rangle\!\langle z_T| g_{T-1}|z_{T+1}\rangle\cdots \langle z_{2(T-1)}|g_1|0^a,x\rangle\\
      &=\frac{1}{2^{h}} \sum_{y,z_j} \langle 0^a,x|q_1^\dagger|z_1\rangle\!\langle z_1| q_2^\dagger|z_2\rangle \cdots \langle z_{T-1}|q_T^\dagger|1,y\rangle\!\langle 1,y|q_T|z_T\rangle\!\langle z_T| q_{T-1}|z_{T+1}\rangle\cdots \langle z_{2(T-1)}|q_1|0^a,x\rangle.
\end{aligned}
\end{equation}

We can then define two {\sf \#P} functions $g$ and $f$ as follows: for $g$ count the number of paths $p\in\{(x,y,z_1,...,z_{2(T+1)})\}$ such that the product 

\begin{align}
    &\langle 0^a,x|q_1^\dagger|z_1\rangle\!\langle z_1| q_2^\dagger|z_2\rangle \cdots \langle z_{T-1}|q_T^\dagger|1,y\rangle\!\langle 1,y|q_T|z_T\rangle\!\langle z_T| q_{T-1}|z_{T+1}\rangle\cdots \langle z_{2(T-1)}|q_1|0^a,x\rangle
\end{align} 
evaluates to $1$ and for $f$ count the number of paths such that it evaluates to $-1$. It then follows that the {\sf GapP} function defined by $r = g-f$ satisfies
 \begin{align}
     \frac{r}{2^h} = \text{Tr}[\mathcal V_x^{r}]\,.
 \end{align}
Next, it follows from Ref.~\cite{bordewich2009approximatecountingquantumcomputation} that we can additively approximate $r$ with $u = 2^{N^*}$ where $N^*=|p|$ is the number of bits needed to describe the paths $p$. Since $|x| = n$, $|y| = n+a -1$, and $|z_j| = n+a$ for all $j$, it follows that the number of bits to describe the path is $N^* =  2T(n+a)-(a+1)$. Now, it follows by approximating $r$ and dividing by $2^h$, we get an additive approximation $\tilde {\mathcal V}_x^{r}$ to $\text{Tr}[\mathcal V_x^{r}]$ with $u = 2^{N^*}/2^h$ we have 
\begin{align}
    N_{\geq c}(\mathcal V_x^{r}) -(1-c)2^n -\epsilon2^{N^*-h}  \leq \text{Tr}[\tilde {\mathcal V}_x^{r}] \leq N_{\geq s}(\mathcal V_x^{r}) +s2^n + \epsilon 2^{N^*-h}
\end{align}
which is equivalent to
\begin{align}
    N_{\geq c}(\mathcal V_x^{r}) -\epsilon '2^{N^*-h}  \leq \text{Tr} [\tilde {\mathcal V}_x^{r}] \leq N_{\geq s}(\mathcal V_x^{r}) + \epsilon'' 2^{N^*-h}
\end{align}
when
\begin{align}
    \epsilon' = (1-c)2^{n-N^*+h} + \epsilon\,,\qquad \epsilon'' = s2^{n-N^*+h} + \epsilon\,.
\end{align}
Since $n-N^*+h = n+a+h+1 -2T(n+a)$ and $h\leq T$, $\epsilon'$ and $\epsilon''$ are always inverse polynomial whenever $\epsilon$ is. Thus there are efficient classical additive approximations for {\sf\#BQP} functions when $u = 2^{2T(n+a) - (a+1)-h  }$.
\end{proof}

\subsection{{\sf \#BQP} and {\sf DQC$_1$}}\label{sec:proofs-DQC1-sharpBQP}

Here we prove that additive-error approximations to a subclass of {\sf \#BQP} functions, namely those whose verifier circuits require only a logarthimic number of ancilla qubits, is {\sf DQC}$_1$-complete with normalization scaling exponential in the dimension of the witness space.

\BPPAisDQC*

\begin{proof}
     In order to show that {\sf BPP}$^A$ $\subseteq$ {\sf DQC}$_1$, we prove there exists a polynomial-time quantum circuit, using only a logarithmic number of ancilla qubits, which simulates $A$. In particular, let $D^{r,(x)}$ be a given verifier circuit for a {\sf \#BQP} relation $r$ on an input $x$ with $w_r(|x|)$ qubits and $a_r(|x|)=O(\log(w_r(|x|)))$ ancilla qubits with completeness and soundness parameters $c(|x|)$ and $s(|x|)$ and an associated verifier operator $\mathcal D_x^{r}$. Then one can always construct a block encoding $U\coloneqq \bra{1}D^{r,(x)}(\ket{0}^{\otimes a_r(|x|)}\otimes\mathbb{I}_W)$, and then make use of Ref.~\cite[Theorem 19]{Gilyen:2019a} to implement a singular value threshold projector on $U$. This allows us to construct a new block encoding $U_\Phi=\bra{1}D_\Phi(\ket{0}\ket{+}\ket{0}^{\otimes a_r(|x|)}\otimes \mathbb{I}_W)$, for some $\Phi\in \mathbb{R}^p$, and where $D_\Phi$ is given by the circuit in Fig.~\ref{fig:inclusion-circuit}. The new encoding $U_\Phi$ now has singular values $\tilde{\varsigma}_i$ which are $\varepsilon$-close to $0$ or $1$ depending on whether the singular values $\varsigma_i$ of $U$ are below or above a threshold $t\in(0,1)$.
\begin{figure}[t!]
\includegraphics[width=\textwidth]{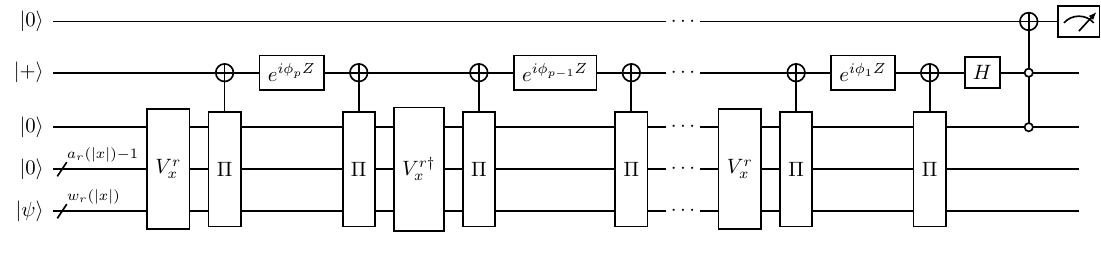}
\caption{\label{fig:inclusion-circuit} The circuit $D_\Phi$ implementing the quantum singular value transformation $U_\Phi = P(U)$ given in Ref.~\cite{Gilyen:2019a} using a degree-$p$, even polynomial approximation. The block encodings $D^{r,(x)}$ and $(D^{r,(x)})^\dagger$ are applied alternately followed by applications of the reflection operators. In particular, controlled on the witness register being in the image of $\Pi$, the ancilla qubit is flipped, then a single qubit Pauli-$Z$ rotation is applied to angles $\phi_i$, where the particular angles are classically determined according to the function being approximated by the polynomial. The final control gate performs an $X$ operation on the top qubit if and only if all of its control qubits are in the $0$ state. The top $a_r(|x|)+2$ qubits are ancillas and the bottom $w_r(|x|)$ qubits are the witness qubits.}
\end{figure}

    To construct the new block encoding $U_\Phi$, define a function on the interval $[-1,1]$ which takes the value $0$ on the interval $[-t,t]$ and $1$ elsewhere. A polynomial approximation $P$ can be constructed such that
    \begin{enumerate}
        \item $|P(x)|\leq 1$ for $x\in[-1,1]$\,, \\
        \item $P(x)\in[1-\varepsilon,1]$ for $x\in[-1,-t-\Delta]\cup[t+\Delta,1]$ with $\Delta\in(0,1)$\,,\\
        \item $P(x)\in[0,\varepsilon]$ for $x\in[-t+\Delta,t-\Delta]$ \,.
    \end{enumerate}
    Additionally, it can be shown that this polynomial approximation is even and has degree $p=O(\log(1/\varepsilon)/\Delta)$ using techniques similar to that of Ref.~\cite[Corollary 16]{Gilyen:2019a}. Finally, according to Ref.~\cite[Corollary 11]{Gilyen:2019a}, given the polynomial approximation $P$, one can construct a vector $\Phi\in\mathbb{R}^p$ which applies $P$ to the singular values of $U$, resulting in the new encoding $U_\Phi= P(U)$. The implementation of $P(U)$ is based on quantum signal processing techniques, and requires only one ancilla qubit with $p$ applications of $D^{r,(x)}$ and $(D^{r,(x)})^\dagger$ and $O(p)$ other elementary gates, as shown in Fig.~\ref{fig:inclusion-circuit}.

    Thus, given the verifier circuit $D^{r,(x)}$, we define the new verifier circuit $D_\Phi$ in Fig.~\ref{fig:inclusion-circuit} that implements the singular value threshold projector protocol on its block encoding $U$, setting $t=(s(|x|)+c(|x|))/2$ and $\Delta=(c(|x|)-s(|x|))/2$, to obtain the new encoding $U_\Phi=\bra{1}D_\Phi(\ket{0}\ket{+}\ket{0}^{\otimes a_r(|x|)}\otimes \mathbb{I}_W)$ with associated acceptance operator
    \begin{equation}
    \begin{aligned}
        \mathcal{D}_\Phi& = U_\Phi^\dagger U_\Phi\,, \\
        &=(\bra{0}\bra{+}\bra{0}^{\otimes a_r(|x|)}\otimes \mathbb{I}_W)D_\Phi^\dagger|1\rangle\langle 1|D_\Phi(\ket{0}\ket{+}\ket{0}^{\otimes a_r(|x|)}\otimes\mathbb{I}_W)\,, \\
        &=\sum_i \tilde{\varsigma}_i\ket{\psi_i}\bra{\psi_i}\,,
    \end{aligned}
    \end{equation}
    where $\ket{\psi_i}$ form an orthonormal basis and $\tilde{\varsigma}_i$ are the amplified singular values such that $\tilde{\varsigma}_i\in[1-\varepsilon,1]$ when $\varsigma_i\geq c(|x|)$ and $\tilde{\varsigma}_i\in[0,\varepsilon]$ when $\varsigma_i\leq s(|x|)$. Due to  the amplification, the trace of this operator satisfies
    \begin{align}
        N_{\geq c(|x|)}(\mathcal D_x^{r})-(2\varepsilon -\varepsilon^2)2^{w_r(|x|)}\leq\text{Tr}[\mathcal D_\Phi]\leq N_{\geq s(|x|)}(\mathcal D_x^{r}) + \varepsilon^22^{w_r(|x|)}\,.
    \end{align}
    
    Finally, the trace of this amplified acceptance operator can be additively estimated to normalization $u(|x|)=2^{w_r(|x|)}$ and precision $\epsilon$ via the same algorithm as in the proof of Theorem \ref{thm:BPPA-BQP} using $O(1/\epsilon^2)$ repetitions of $ D_\Phi$. This outputs an estimate $\tilde X$ satisfying
    \begin{align}
        N_{\geq c(|x|)}(\mathcal D_x^{r})-(2\varepsilon -\varepsilon^2 +\epsilon)2^{w_r(|x|)}\leq \tilde X \leq N_{\geq s(|x|)}(\mathcal D_x^{r}) + (\varepsilon^2+
        \epsilon)2^{w_r(|x|)}\,,
    \end{align}
    with probability greater than $3/4$.
    Thus by taking $\varepsilon=O(\log(\text{poly}(\epsilon')^{-1}))$ such that both $2\varepsilon -\varepsilon^2$ and $\varepsilon^2$ are less than or equal to $\epsilon$, an additive error estimate to normalization $2^{w_r(|x|)}$ and precision $\epsilon'=2\epsilon$ can be obtained. Then, noting that $\Delta = (c(|x|) - s(|x|))/2 = \Omega(\text{poly}(|x|)^{-1})$, an additive error estimate to normalization $u(|x|) = 2^{w(|x|)}$ and to precision $\epsilon'$ can be obtained using $O(\text{poly}\left(1/\epsilon'\right))$ runs of $O(\log(\text{poly}(1/\epsilon'))\cdot\text{poly}(|x|))$ size quantum circuits. Then repeating this procedure $\log(1/\delta')$ times and taking the median of estimates can improve the success probability to $1-\delta'$. Importantly since the circuits used in this algorithm require only a logarithmic overhead of $a_r(|x|)+2 = O(\log(w_r(|x|)))$ pure ancilla qubits, they can be computed using {\sf DQC}$_1$ circuits.
    Therefore {\sf BPP}$^A$ $\subseteq$ {\sf DQC}$_1$.

    On the other hand, to show that {\sf DQC}$_1$ $\subseteq$ {\sf BPP}$^A$, we use similar techniques as those in the proof of Theorem~\ref{thm:BPPA-BQP}. We show that querying $A$ for a {\sf \#BQP}$_{(c^r,s^r)}$ relation $r$ with $a_r(|x|) = O(\log(w_r(|x|)))$ ancilla qubits on a polynomial number of inputs $x_i$ to additive error, with normalizations $u_i=2^{w(|x_i|)}$, allows us to solve any {\sf DQC}$_1$ decision problem as given in Definition~\ref{def:dqc1}. In fact, the Average Accept Probability for Verifier Circuits problem in Definition~\ref{def:avg-accp-verifier-circ} reduces to the definition of {\sf DQC}$_1$ in Definition~\ref{def:dqc1} when only a logarithmic number of ancilla qubits are allowed. A subtle difference here is that we should like the circuits used to be generated on an input-by-input basis, as is done in the definition of {\sf DQC}$_1$. Rather than starting from Definition~\ref{def:sharp-BQP-new} as before, we modify it to allow for input-dependent verifier-circuit generation, with the acceptance operator $\mathcal{D}_x^{r}$ corresponding to the circuit $D^{r,(x)}$ for relation $r$ and input $x$.

    Establishing this fact, the remainder of the proof is identical to the proof of Theorem~\ref{thm:BPPA-BQP}, replacing $\mathcal{V}_x^{r}$ with $\mathcal{D}_x^{r}$. First, partition the interval $[0,1]$ into $M=\mathcal{O}(\text{poly}(|x|))$ subintervals $[s_i,c_i]$, each of size $\Omega(\text{poly}(|x|)^{-1})$, corresponding to the different completeness and soundness parameters specified by the fixed-length prefix-free encoding $y_i$. Specifically, we consider the subintervals $[s_i,c_i]$ where $c_i=(M-i)/M$ and $s_i=c_i-1/4M$ for $1\leq i\leq M-1$. Then querying $A$ for each problem input $x_i$ returns an estimate
    \begin{equation}
        \hat{n}_i=N_{\geq c_i}+\delta_i+\varepsilon_i\,,
    \end{equation}
    with $\delta_i\in[0,N_{[s_i,c_i]}]$ and $|\varepsilon_i|\in[0,\epsilon 2^{w(|x_i|)}]$. Letting $\hat{n}_0=0$, we define the two quantities $\hat{\chi}_i=\hat{n}_i-\hat{n}_{i-1}$ and $\Delta_i=c_i+1/2M$. Then we have
    \begin{equation}
    \begin{aligned}
        \sum_{i=1}^M \Delta_i\hat{\chi}_i&=\sum_{i=1}^M \Delta_i\left(N_{[c_i,c_{i-1}]}+\delta_i-\delta_{i-1}+\varepsilon_i-\varepsilon_{i-1}\right)\,, \\
        &=\sum_{i=1}^M \Delta_iN_{[c_i,c_{i-1}]}+\sum_{i=1}^{M-1}(\delta_i+\varepsilon_i)+\underbrace{\frac{1}{2M}(\delta_M+\epsilon_M)}_{=0}\,.
    \end{aligned}
    \end{equation}
    where the last term vanishes as a result of the fact that for the interval $[s_M,c_M]$, we need not query the oracle and instead always return $2^{w(|x|)}$, meaning $\delta_M,\epsilon_M=0$.
    
    This quantity approximates the trace of the associated acceptance operator $\mathcal{D}_{x}^{r}$ of {\sf DQC}$_1$ in Definition~\ref{def:dqc1}. In particular, note that
    \begin{equation}
        |\Delta_i N_{[c_i,c_{i-1}]}-\text{Tr}_{[c_i,c_{i-1}]}(\mathcal{D}_{x}^{r})|\leq \frac{1}{2M}N_{[c_i,c_{i-1}]}\,,
    \end{equation}
    which then implies that 
    \begin{equation}
        \bigg|\sum_{i=1}^M\Delta_iN_{[c_i,c_{i-1}]}-\text{Tr}(\mathcal{D}_{x}^{r})\bigg|\leq \frac{2^{w(|x|)}}{2M}\,.
    \end{equation}

    Finally, we bound the error resulting from $\delta_i$ and $\varepsilon_i$ by noting that
    \begin{equation}
        \frac{1}{M}\sum_{i=1}^M \delta_i\leq \frac{2^{w(|x|)}}{M},\quad \text{and}\quad\frac{1}{M}\sum_{i=1}^M \varepsilon_i \leq \epsilon 2^{w(|x|)}\,,
    \end{equation}
    where the first inequality holds from the fact that $[s_i,c_i]\cap[s_j,c_j]=\emptyset$ for all $i\neq j$. In the second inequality we have used the fact that $2^{w(|x_i|)}=2^{w(|x|)}$ for all $i$, or rather, the size of the witness subspace for each input $x_i$ is the same size as that of $x$ even though the classical inputs may vary in size. Setting $\epsilon=1/M$, we add the errors from each contribution and find that
    \begin{equation}
        \bigg|\sum_{i=1}^M\Delta_i\hat{\chi}_i-\text{Tr}(\mathcal{D}_x^{r})\bigg|\leq \frac{2^{w(|x|)}}{M}+\frac{2^{w(|x|)}}{M}+\frac{2^{w(|x|)}}{2M}=\frac{5}{2}\frac{2^{w(|x|)}}{M}\,.
    \end{equation}
    Choosing $M>5/(c(|x|)-s(|x|))$ suffices to distinguish between the accepting and rejecting cases of Definition~\ref{def:dqc1}. Therefore, {\sf DQC}$_1$ $\subseteq$ {\sf BPP}$^A$.
\end{proof}

\section{Discussion}\label{sec:discussion}

We can now comment on the implications of our results in the context of previous work. First, recall that the quantum density of states problem is {\sf \#BQP}-complete~\cite{Brown_2011}, and it is well known that returning additive-error approximations to this problem, for a particular normalization $u=2^n$, is {\sf DQC$_1$}-complete \cite{brandao2008entanglement,cade2017quantum}. A natural question to ask is whether this result can be understood in terms of the more general framework we provide. Here we elucidate how the discrepancy in complexity arises between additive-error approximations to the quantum density of states problem and additive-error approximations to {\sf \#BQP} functions. The essential component of this difference, illustrated in our results, is determining how the size of the witness space is affected by the explicit reduction given in Ref.~\cite{Brown_2011}.

The {\sf \#BQP}-hardness reduction for the quantum density of states problem~\cite{Brown_2011} takes an input $x$ to a {\sf\#BQP} relation $r$ with an associated verifier circuit $V^{r,(|x|)}$, and employs the Feynman-Kitaev circuit-to-Hamiltonian mapping to construct a Hamiltonian whose low-energy subspace corresponds to the subspace of witnesses that accept with high probability. 

If $V^{r,(|x|)}$ has $w_r(|x|)$ witness qubits, $a_r(|x|)$ ancilla qubits, and $t_r(|x|)$ gates, then the smallest known Hamiltonian that encodes the accepting witnesses in its ground space is a Hamiltonian on $w_r(|x|)+a_r(|x|) +O(\log(t(|x|)))$ qubits. This construction has the drawback of logarithmic locality, which is necessary to achieve $O(\log(t(|x|)))$, rather than $t(|x|)$, scaling from the gate count. Since the {\sf DQC$_1$}-completeness of returning additive-error approximations requires a normalization $u=2^n$ for an $n$ qubit Hamiltonian, after this reduction the normalization achievable via quantum algorithms will increase from $u=2^{w(|x|)}$ to $u=2^{w(|x|)+a(|x|)+O(\log(t(|x|)))}$. If we also consider the case where $a(|x|) = O(\log(w(|x|)))$, then $2^{a(|x|)+O(\log(t(|x|)))} = O(\text{poly}(|x|))$, and we can compensate for this polynomially larger prefactor by choosing a polynomially smaller precision $\epsilon$. However, in the case where $a(|x|)$ scales superlogarithmically, the increase in the normalization will scale superpolynomially, and such a prefactor can no longer be compensated for by choosing a smaller precision $\epsilon$. This observation naturally distinguishes the two cases of {\sf \#BQP} relations considered in this work. The {\sf \#BQP} relations which require only logarithmically many ancilla qubits have reductions to and from the quantum density of states problem which preserve, up to a polynomially large prefactor, the normalization $u$ of the additive error achievable. Alternatively, for general {\sf \#BQP} relations the reductions highly impact the quality of the achievable normalization.            

We also note that the superpolynomial increase of the achievable normalization for general {\sf \#BQP} relations is also why the log-local versions of density of states and quantum partition function problem are {\sf DQC$_1$}-complete. As indicated in the preceding paragraph, a circuit-to-Hamiltonian mapping with constant locality requires a Hamiltonian which acts on at least $t(|x|)$ more qubits than the original {\sf DQC$_1$} circuit. This increases the normalization enough that the {\sf DQC$_1$} problem can no longer be decided.    

Another connection which can be made more straightforward as a result of our framework is the comparison between the various types of approximations to counting functions, and in particular, how they can lead to vastly different complexities. Recall that returning exact solutions to various counting problems yields little difference in complexity. In particular, it can be shown that both {\sf GapP} and {\sf \#BQP} are polynomial-time reducible to a \textit{single} oracle query for {\sf \#P}~\cite{counting-O,Brown_2011}.

However, when considering approximate solutions to various counting functions, this polynomial-time equivalence breaks down. In the case of the relative-error approximations, {\sf \#P} is approximated by a randomized algorithm that queries an {\sf NP} oracle, placing it in a low level of the polynomial hierarchy \cite{Stockmeyer1983TheCO}. On the other hand, relative-error approximations to {\sf GapP} remain {\sf \#P}-hard. The reason for this is that a relative-error approximation preserves the sign of a quantity, and this information can be used to approach the exact value via binary search. Finally, for {\sf \#BQP} the complexity of returning relative-error approximations is an open problem, with the complexity falling somewhere between {\sf QMA} and {\sf \#P} \cite{Bravyi_2022,bravyi2024kronecker}. 

For the additive-error approximations we consider here, the difference in complexity between {\sf \#P} and {\sf GapP} vanishes as the approximations no longer reliably identify the sign of the functions. An interesting property of {\sf \#BQP} is that it appears to be more easily approximated by quantum algorithms than by classical algorithms. This is a result of the fact that for exact solutions to counting functions, the polynomial-time equivalence between {\sf \#BQP}, {\sf \#P}, and {\sf GapP} highly affects the size of the witnesses for the respective problems. As discussed at the beginning of this section, this leads to vastly different normalizations of the efficiently achievable approximations.

\section{Conclusion}\label{sec:conclusion} 

In this work, we focus primarily on additive-error approximations to {\sf\#BQP} relations. While exact solutions to counting functions are theoretically interesting, they are generally difficult to compute, but approximate solutions to counting functions can often be efficiently obtained and are thus more useful for relevant applications. Here, we proved the existence of efficient quantum algorithms for returning additive-error approximations to {\sf \#BQP} relations with a normalization $u$ which scales exponentially in the number of witness qubits of the associated verifier circuit. We also demonstrated that returning such approximations is {\sf BQP}-hard, providing evidence that such approximations are not efficiently achievable classically. The bounds on the normalization we provide for efficient approximations are also shown to be, in a sense, optimal, as we proved that for any smaller normalization, returning additive-error approximations is {\sf \#BQP}-hard.

We also explored the connection between returning additive-error approximations to {\sf \#BQP} relations and the complexity class {\sf DQC$_1$}. Here we proved that returning additive-error approximations, with a normalization scaling exponentially in the size of the witness space, to {\sf \#BQP} relations requiring only logarithmically many ancilla qubits, is {\sf DQC$_1$}-complete.    

Finally, we discussed the implications of our results in the context of previously known results. Given that finding approximate solutions to various counting functions leads to vastly different complexities in the relative-error regime, and, moreover, that the complexity of such approximations to {\sf \#BQP} is still unclear, it motivates the study of other types of approximations to {\sf \#BQP}. In particular, we identify two regimes of {\sf \#BQP} problems which yield efficient additive-error solutions, and this division is solely dependent on the number of ancilla qubits, whether polynomial or logarithmic, needed for the verifier circuits.

\section*{Acknowledgements}
This material is based upon work supported by the U.S. Department of Energy, Office of Science, National Quantum Information Science Research Centers, Quantum Science Center (QSC) and Quantum Systems Accelerator (QSA). Mason L. Rhodes was a participant in the 2020 Quantum Computing Summer School at LANL, sponsored by the LANL Information Science \& Technology Institute. Samuel Slezak and Yigit Subasi were supported by the Laboratory Directed Research and Development program of Los Alamos National Laboratory (LANL) under project number 20230049DR. Samuel Slezak and Yigit Subasi were funded by the QSC and Mason L. Rhodes was funded by the QSA to perform the analytical calculations and to write the manuscript along with the other authors.

Sandia National Laboratories is a multimission laboratory managed and operated by National Technology and Engineering Solutions of Sandia, LLC., a wholly owned subsidiary of Honeywell International, Inc., for the U.S.\ Department of Energy's National Nuclear Security Administration under contract DE-NA-0003525.

This paper describes objective technical results and analysis. Any
subjective views or opinions that might be expressed in the paper do not
necessarily represent the views of the U.S.\ Department of Energy or the
United States Government.

\bibliographystyle{plain}
\bibliography{refs}

\section{Equivalence of {\sf \#BQP} definitions \label{app:sharp-bqp-equiv}}

Here we prove Lemma \ref{lem:equiv-sharp-QP-defs} demonstrating the equivalence of Definitions~\ref{def:sharp-BQP-new} and~\ref{def:sharp-BQP}. The original definition imposes that the completeness and soundness parameters can only depend on the size of the input, whereas the alternative definition allows the completeness and soundness to depend on the input itself. We demonstrate their equivalence in order to ensure that Definition~\ref{def:sharp-BQP-new} is not more powerful than the original one in Definition~\ref{def:sharp-BQP}.

\begin{restatable}{lem}{equivdefs}\label{lem:equiv-sharp-QP-defs}
    ${\sf \#BQP}_{(c^r,s^r)} = {\sf \#BQP}$.
\end{restatable}

\begin{proof}
    Let $c(x)$ and $s(x)$ be the completeness and soundness parameters associated to {\sf \#BQP}$_{(c^r,s^r)}$ and let $c'(|x|)$ and $s'(|x|)$ be the completeness and soundness parameters of {\sf \#BQP}. The inclusion {\sf \#BQP} $\subset$ {\sf \#BQP$_{(c^r,s^r)}$} follows immediately from the definitions, setting $c(x)=c'(|x|)$ and $s(x)=s'(|x|)$. For the inclusion {\sf \#BQP$_{(c^r,s^r)}$} $\subset$ {\sf \#BQP}, given the family of circuits $V^{r,(|x|)}$ and functions $c(x)$ and $s(x)$ for the relation $r$, we can construct a degree-$m$ even polynomial approximation $P$ to the rectangle function with $m=O(\log(1/\epsilon)/\delta)$~\cite{Gilyen:2019a} just as in the proof of Theorem~\ref{thm:BPPA-DQC1}. Specifically, $\epsilon$ determines how much the promise gap is amplified such that we can choose $\epsilon=\Omega(\text{exp}(|x|)^{-1})$ (to be specified later) and $\delta$ corresponds to the uncertainty around the threshold for implementing the threshold projectors, and thus must satisfy $\delta\leq (c(x)-s(x))/2=\Omega(\text{poly}(|x|)^{-1})$. Given such a polynomial, there exists a vector $\Phi(x)\in\mathbb{R}^m$ which applies the polynomial to the singular values of $\bra{1}V^{r,(|x|)}(\ket{x}\otimes|0^{a'}\rangle\otimes\mathbb{I}_W)$, resulting in a new family of circuits $V_{\Phi(x)}^{r,(|x|)}$ such that 
    $$
    \mathcal V^{r}_{\Phi(x)} = (\bra{x}\otimes\langle0^{a'}|\otimes\mathbb{I}_W)(V^{r,(|x|)}_{\Phi(x)})^\dagger \Pi_{\text{out}}V^{r,(|x|)}_{\Phi(x)}(\ket{x}\otimes|0^{a'}\rangle\otimes\mathbb{I}_W)\,,
    $$
    where we now choose $\epsilon$ such that all of the eigenstates of $\mathcal V^{r}_{x}$ with eigenvalues above $c(x)$ are now eigenstates of $\mathcal V^{r}_{\Phi(x)}$ with eigenvalues above ($1-\epsilon)^2>c'(|x|)$, and all of the eigenstates of $\mathcal V^{r}_{x}$ with eigenvalues below $s(x)$ are now eigenstates of $\mathcal V^{r}_{\Phi(x)}$ with eigenvalues below $\epsilon^2<s'(|x|)$.
\end{proof}

\section{{\sf BQP}-completeness of the Average Accept Probability of Verifier Circuit Problem}\label{app:avg-accp-verifer-bqp-complete}

Here we prove Lemma~\ref{lem:avg-accept-verifer-circ-BQP-comp}, showing that the Average Accept Probability of Verifier Circuits Problem is {\sf BQP}-complete. We point out that this problem is very similar to the definition (and thus what would be a defining problem for) {\sf DQC$_1$}. Indeed, if the circuits were generated on an input by input basis and the number of ancilla qubits were only allowed to scale as $O(\log(|x|))$, then this would be a trivially {\sf DQC$_1$} complete problem.

\begin{restatable}{lem}{avgaccept}\label{lem:avg-accept-verifer-circ-BQP-comp}
 The Average Accept Probability of Verifier Circuits problem in Definition~\ref{def:avg-accp-verifier-circ} is {\sf BQP}-complete.
\end{restatable}

\begin{proof}
    In order to demonstrate hardness, assume we are given an input $x$ to a problem in {\sf BQP} where we will take the completeness and soundness parameters to be $c(|x|)=2/3$ and $s(|x|)= 1/3$. We can then construct the corresponding circuit $V^{(|x|)}$ of the appropriate size and consider its acceptance operator $\mathcal V_x$ as in Eq.~\eqref{eq:gen-acceptance-op}. Note that this acceptance operator has zero witness qubits, and is thus a number (which can also be thought of as a one dimensional operator). We can then solve the average accept probability to determine whether:
    \begin{align}
        \frac{\text{Tr}[\mathcal V_x]}{2^0} = \bra{x, 0^a}V^{(|x|)\dagger}\Pi_{\textup{out}}V^{(|x|)}\ket{x, 0^a} 
    \end{align}
    is greater than $2/3$ or less than $1/3$, which clearly decides the {\sf BQP} problem.

    For the inclusion, we use a similar algorithm as in Theorem~\ref{thm:BPPA-BQP}. Given an input $x$, circuit $V^{(|x|)}$ acting on $w(|x|)$ witness qubits and $a(|x|)$ ancilla qubits, randomly generate $M>3/ \epsilon^2$ $w(|x|)$-bit strings and send them through the witness subsystem of the circuit, incrementing a counter when the circuit accepts to build the estimator
    \begin{align}
         X = \frac{1}{M} \sum_{i=1}^MX_i.
    \end{align}
    where $X_i$ is an indicator function for the $i^{\text{th}}$ run accepting. 
    With the stated $M$, this estimator will return a value $\tilde X$ that satisfies
    \begin{align}\label{eq:QMA-BQP-estimator-event}
        \left|\tilde X -\frac{\text{Tr}[\mathcal V_x]}{2^{w(|x|)}}\right| \leq \epsilon
    \end{align}
with probability larger than $2/3$. Choosing $\epsilon < \frac{1}{6}$ allows the decision problem to be determined whenever the event Eq.~\eqref{eq:QMA-BQP-estimator-event} is satisfied, and the algorithm will output yes if $\tilde X$ is closer to or greater than $c$, and will output no when $\tilde X$ is closer to or smaller than $s$. Thus it follows that if the instance is a yes instance this procedure will accept with probability greater than $2/3$, and will accept with probability less than $1/3$ for no instances. 
\end{proof}

\end{document}